\documentclass[journal, twocolumn, 10pt]{IEEEtran}

\usepackage{graphicx}
\usepackage{cite}
\usepackage{amsmath, amsthm, amssymb, bm}
\usepackage{amsthm}
\usepackage{algorithm}
\usepackage{algpseudocode}
\usepackage{epstopdf}
\usepackage{dsfont}
\usepackage{color}
\usepackage{subcaption}

\newcommand{\skipval}{0.87mm} 
\newtheorem{theorem}{Theorem}[section]
\newtheorem{property}{Property}[section]

\usepackage{algorithm}

\long\def\comment#1{}

\newfont{\bbb}{msbm10 scaled 700}

\newfont{\bb}{msbm10 scaled 1100}

\newcommand{\bv}{{\bf b}}

\newcommand{\xv}{{\bf x}}
\newcommand{\yv}{{\bf y}}
\newcommand{\zv}{{\bf z}}

\newcommand{\Hm}{{\bf H}}
\newcommand{\Id}{{\bf I}}

\newcommand{\Um}{{\bf U}}

\newcommand{\Sigmam}{\hbox{\boldmath$\Sigma$}}

\ifCLASSINFOpdf
\else
\fi
\begin{document}
%

\title{Supplementary Appendix for: Constrained Perturbation Regularization Approach for Signal Estimation Using Random Matrix Theory}
%
%

        



\author{Mohamed~Suliman, Tarig~Ballal, Abla~Kammoun, and Tareq Y. Al-Naffouri


}

\maketitle

\begin{abstract}
In this supplementary appendix we provide proofs and additional extensive simulations that complement the analysis of the main paper (constrained perturbation regularization approach for signal estimation using random matrix theory).
\end{abstract}

\section{Properties of the COPRA function $S\left(\tilde{\gamma}_{\text{o}}\right)$}
\label{sec3:Properties}

Firstly, we discuss the general properties of the COPRA function
\begin{align}
\label{eq:BPR sec}
S\left(\tilde{\gamma_{\text{o}}}\right)  &= \text{Tr}\left(\Sigmam^{2}\left(\Sigmam^{2}+N\tilde{\gamma_{\text{o}}}\Id\right)^{-2}\bv\bv^{H}\right) \Big[\delta_{\text{o}}^{2} \tilde{\delta}_{\text{o}}^{2} - \tilde{\gamma_{\text{o}}}^{2} \delta_{\text{o}} \nonumber\\
& - \tilde{\gamma_{\text{o}}} \delta_{\text{o}} \tilde{\delta}_{\text{o}} \Big] +\text{Tr}\left(\left(\Sigmam^{2}+N\tilde{\gamma_{\text{o}}}\Id\right)^{-2}\bv\bv^{H}\right) \times \nonumber\\
&\Big[( N  \delta_{\text{o}} \tilde{\delta}_{\text{o}}  \left( \tilde{\gamma_{\text{o}}}^{2} -  \tilde{\gamma_{\text{o}}} \delta_{\text{o}} \tilde{\delta}_{\text{o}} -  \delta_{\text{o}} \tilde{\delta}_{\text{o}}^{2} \right) +M \tilde{\delta}_{\text{o}} \nonumber\\
& \tilde{\gamma_{\text{o}}} \left(\tilde{\gamma_{\text{o}}}- \tilde{\gamma_{\text{o}}}\delta_{\text{o}} +  \delta_{\text{o}}^{2} \tilde{\delta}_{\text{o}} \right) \Big] = 0,
\end{align}
where $\bv \triangleq \Um^{H} \yv$. 

The COPRA characteristic equation is a function of the problem parameters, the received signal, and the unknown regularizer $\tilde{\gamma}_{\text{o}}$. Our main interest is to find a positive root for (\ref{eq:BPR sec}). Before that, and to simplify the properties analysis of the COPRA function, we will assume that $M/N \approx 1$\footnote{This assumption is to simplify the properties analysis presentation. However, all the properties and the theorems that will be presented and proved in this section can be easily extended for a general ratio of $M/N$ upon following the same steps.}. Thus, (\ref{eq:BPR sec}) can be written as
\begin{align}
\label{eq3:R-BPR sec}
&S\left(\tilde{\gamma}_{\text{o}}\right) =
\text{Tr}\left(\Sigmam^{2}\left(\Sigmam^{2}+N\tilde{\gamma}_{\text{o}}\Id\right)^{-2}\bv\bv^{H}\right)\times \nonumber\\
&\Bigg[\tilde{\gamma}_{\text{o}} \left(\sqrt{\frac{\tilde{\gamma}_{\text{o}}+4}{\tilde{\gamma}_{\text{o}}}}-1\right)-4\Bigg] + 
\text{Tr}\left(\left(\Sigmam^{2}+N\tilde{\gamma}_{\text{o}}\Id\right)^{-2}\bv\bv^{H}\right) \times \nonumber\\
&\Bigg[ N\tilde{\gamma}_{\text{o}} \left(\left(\sqrt{\frac{\tilde{\gamma}_{\text{o}}+4}{\tilde{\gamma}_{\text{o}}}}-1\right) \tilde{\gamma}_{\text{o}}+2 \sqrt{\frac{\tilde{\gamma}_{\text{o}}+4}{\tilde{\gamma}_{\text{o}}}}-4\right) \Bigg] = 0. 
\end{align}

Before proceeding further, we will define the following two functions 
\begin{align}
\label{eq3:R-BPR1 sec}
S_{1}\left(\tilde{\gamma}_{\text{o}}\right) & \triangleq
\text{Tr}\left(\Sigmam^{2}\left(\Sigmam^{2}+N\tilde{\gamma}_{\text{o}}\Id\right)^{-2}\bv\bv^{H}\right) \times \nonumber\\
&\Bigg[\tilde{\gamma}_{\text{o}} \left(\sqrt{\frac{\tilde{\gamma}_{\text{o}}+4}{\tilde{\gamma}_{\text{o}}}}-1\right)-4\Bigg],
\end{align}
and 
\begin{align}
\label{eq3:R-BPR2 sec}
 &S_{2}\left(\tilde{\gamma}_{\text{o}}\right) \triangleq
\text{Tr}\left(\left(\Sigmam^{2}+N\tilde{\gamma}_{\text{o}}\Id\right)^{-2}\bv\bv^{H}\right)\times \nonumber\\
&\Bigg[ N\tilde{\gamma}_{\text{o}} \left(\left(\sqrt{\frac{\tilde{\gamma}_{\text{o}}+4}{\tilde{\gamma}_{\text{o}}}}-1\right) \tilde{\gamma}_{\text{o}}+2 \sqrt{\frac{\tilde{\gamma}_{\text{o}}+4}{\tilde{\gamma}_{\text{o}}}}-4\right) \Bigg].\nonumber\\
\end{align}

Now, let us start our properties discussion by examine some main properties of the COPRA function that are straightforward to proof.

\begin{property}
\label{pro3:p11}
The function $S\left(\tilde{\gamma}_{\text{o}}\right)$ has $M$ discontinuities at $\tilde{\gamma}_{\text{o}} = -\sigma_{i}^{2}/N, \forall i =1,\dotsi,M$. However, these discontinuities are of no interest as far as COPRA is concerned.
\end{property}

\begin{property}
\label{pro3:p3}
$S\left(\tilde{\gamma}_{\text{o}}\right)$ is continuous over the interval $\left(0, +\infty\right)$.
\end{property}

\begin{property}
\label{pro3:p1}
$\lim_{\tilde{\gamma}_{\text{o}} \to +\infty} S\left(\tilde{\gamma}_{\text{o}}\right) = 0$.
\end{property}

\begin{property}
\label{pro3:p2}
$\lim_{\tilde{\gamma}_{\text{o}} \to 0^{+}} S\left(\tilde{\gamma}_{\text{o}}\right)=-4 \ \text{Tr}\left(\Sigmam^{-2}\bv\bv^{H}\right)$.
\end{property}

\begin{proof}
Let $\Sigmam^{2}= \text{diag}\left(\sigma_{1}^{2}, \sigma_{2}^{2}, \dotsi,\sigma_{M}^{2}\right)$ and $\bv \bv^{H} = \text{diag}\left(b_{1}^{2}, b_{2}^{2}, \dotsi,b_{M}^{2}\right)$. Then, (\ref{eq3:R-BPR sec}) can be written as\footnote{Since $\bv \bv^{H}$ is multiplied by a diagonal matrix $\Sigmam^{2}$ inside the trace, we can only consider its diagonal  entries.} 

\begin{align}
\label{eq3:p2_1} 
 S\left(\tilde{\gamma}_{\text{o}}\right) &=\sum_{i=1}^{M}\frac{\sigma_{i}^{2} b_{i}^{2}}{\left(\sigma_{i}^{2}+N\tilde{\gamma}_{\text{o}}\right)^{2}}\Bigg[-4-\tilde{\gamma}_{\text{o}}+ \tilde{\gamma}_{\text{o}}\sqrt{\frac{\tilde{\gamma}_{\text{o}}+4}{\tilde{\gamma}_{\text{o}}}} \ \Bigg] \nonumber\\
&+
\sum_{i=1}^{M}\frac{ b_{i}^{2}}{\left(\sigma_{i}^{2}+N\tilde{\gamma_{\text{o}}
 }\right)^{2}} \Bigg[-4N\tilde{\gamma}_{\text{o}}+2 N\tilde{\gamma}_{\text{o}} \sqrt{\frac{\tilde{\gamma}_{\text{o}}+4}{\tilde{\gamma}_{\text{o}}}} \nonumber\\
 &-N\tilde{\gamma}_{\text{o}}^{2} + N\tilde{\gamma}_{\text{o}}^{2}\sqrt{\frac{\tilde{\gamma}_{\text{o}}+4}{\tilde{\gamma}_{\text{o}}}} \ \Bigg],
\end{align} 
which after some algebraic manipulations yields
\begin{align}
\label{eq3:p2_2} 
 S\left(\tilde{\gamma}_{\text{o}}\right) &= \sum_{i=1}^{M}\frac{\sigma_{i}^{2} b_{i}^{2}}{\left(\sigma_{i}^{2}+N\tilde{\gamma}_{\text{o}}\right)^{2}}\Bigg[-4-\tilde{\gamma}_{\text{o}}+ \sqrt{\tilde{\gamma}_{\text{o}}}\sqrt{\tilde{\gamma}_{\text{o}}+4}\Bigg] \nonumber\\
&+
\sum_{i=1}^{M}\frac{ b_{i}^{2}}{\left(\sigma_{i}^{2}+N\tilde{\gamma_{\text{o}}
 }\right)^{2}} \Bigg[-4N\tilde{\gamma}_{\text{o}}+2 N\sqrt{\tilde{\gamma}_{\text{o}}}\sqrt{\tilde{\gamma}_{\text{o}}+4} \nonumber\\
 &- N\tilde{\gamma}_{\text{o}}^{2} + N\tilde{\gamma}_{\text{o}} \sqrt{\tilde{\gamma}_{\text{o}}}\sqrt{\tilde{\gamma}_{\text{o}}+4} \Bigg].
\end{align}
Now, taking the limit as $\tilde{\gamma}_{\text{o}} \to 0^{+}$ yields
\begin{align} 
\label{eq3:p2_3} 
\lim_{\tilde{\gamma}_{\text{o}} \to 0^{+}} S\left(\tilde{\gamma}_{\text{o}}\right) = -4 \sum_{i=1}^{M}\frac{ b_{i}^{2}}{\sigma_{i}^{2}}  =-4 \text{Tr}\left(\Sigmam^{-2}\bv\bv^{H}\right).
\end{align}
\end{proof}

\begin{property}
\label{pro3:p4}
$\lim_{\tilde{\gamma}_{\text{o}} \to +\infty} S\left(\tilde{\gamma}_{\text{o}}\right)$ approaches zero from the negative region.
\end{property}

\begin{proof}
Starting from (\ref{eq3:p2_2}), we can write this equation as 
\begin{align} 
\label{ApC:eq1} 
 S\left(\tilde{\gamma_{\text{o}}}\right) &= \frac{1}{N^{2}\tilde{\gamma_{\text{o}}}^{2}}\sum_{i=1}^{M}\frac{\sigma_{i}^{2} b_{i}^{2}}{\left(\frac{\sigma_{i}^{2}}{N\tilde{\gamma_{\text{o}}}}+1\right)^{2}}\Bigg[-4-\tilde{\gamma_{\text{o}}}+ \sqrt{\tilde{\gamma_{\text{o}}}}\sqrt{\tilde{\gamma_{\text{o}}}+4}\Bigg] \nonumber\\
&+
\frac{1}{N^{2}\tilde{\gamma_{\text{o}}}^{2}} \sum_{i=1}^{M}\frac{ b_{i}^{2}}{\left(\frac{\sigma_{i}^{2}}{N\tilde{\gamma_{\text{o}}}}+1\right)^{2}} \Bigg[-4N\tilde{\gamma_{\text{o}}}+2 N\sqrt{\tilde{\gamma_{\text{o}}}}\sqrt{\tilde{\gamma_{\text{o}}}+4} \nonumber\\\
&-N\tilde{\gamma_{\text{o}}}^{2} + N\tilde{\gamma_{\text{o}}} \sqrt{\tilde{\gamma_{\text{o}}}}\sqrt{\tilde{\gamma_{\text{o}}}+4} \Bigg].
\end{align}
Now, evaluating the limit of (\ref{ApC:eq1}) as $\tilde{\gamma_{\text{o}}}$ approaches $+\infty$ yields
\begin{align}
\label{ApC:eq2} 
&\lim_{\tilde{\gamma_{\text{o}}} \to +\infty} S\left(\tilde{\gamma_{\text{o}}}\right) = \lim_{\tilde{\gamma_{\text{o}}} \to +\infty} \frac{1}{N^{2}\tilde{\gamma_{\text{o}}}}\Big[\frac{-4}{\tilde{\gamma_{\text{o}}}}  \sum_{i=1}^{M}\sigma_{i}^{2} b_{i}^{2}  -2N \sum_{i=1}^{M} b_{i}^{2}\Big]. \nonumber\\
\end{align}
It can be clearly seen that the limit in (\ref{ApC:eq2}) is equal to zero. However, the COPRA function approaches the zero from the negative direction. 
\end{proof}
\begin{property}
\label{pro3:p5}
The functions $S_{1}\left(\tilde{\gamma}_{\text{o}}\right)$ in (\ref{eq3:R-BPR1 sec}) and $S_{2}\left(\tilde{\gamma}_{\text{o}}\right)$ in (\ref{eq3:R-BPR2 sec}) are completely monotonic in the interval $\left(0 ,+\infty\right)$.
\end{property}

\begin{proof}
According to \cite{feller2008introduction, widder2015laplace}, a function $F\left(\tilde{\gamma}_{\text{o}}\right)$ is completely monotonic if it satisfies

\begin{align}
\label{eq3:completeMonotoneCond}
\left(-1\right)^{n} F^{\left( n \right)}\left(\tilde{\gamma}_{\text{o}}\right) \geq 0, \ 0 <\tilde{\gamma}_{\text{o}} < \infty ,  \forall n \in \mathbb{N},
\end{align}
where $(.)^{\left( n\right)}$ is the $n$'th derivative of the function. \\ 
By continuously differentiating $S_{1}\left(\tilde{\gamma}_{\text{o}}\right)$ and $S_{2}\left(\tilde{\gamma}_{\text{o}}\right)$ we can easily show that the two functions are satisfying the monotonic condition in (\ref{eq3:completeMonotoneCond}).
\end{proof}
\begin{theorem}
\label{theo3:th1}
The COPRA function $S\left(\tilde{\gamma}_{\text{o}}\right)$ in (\ref{eq3:R-BPR sec}) has at most two roots in the interval $\left(0, +\infty \right).$
\end{theorem}

\begin{proof}
The proof of Theorem~\ref{theo3:th1} will be conducted in two steps.  Firstly, it has been shown in \cite{kammler1976chebyshev, kammler1979least} that any completely monotonic function can be approximated as a sum of exponential functions. That is, if $F\left(\tilde{\gamma}_{\text{o}}\right)$ is a completely monotonic function, it can be approximated as 
\begin{align}
\label{eq3:completeMonotoneApprox}
F\left(\tilde{\gamma}_{\text{o}}\right) \approx \sum_{i=1}^{l} a_i e^{-k_{i} \tilde{\gamma}_{\text{o}}},
\end{align}
where $l$ is the number of the terms in the sum. It has been shown that a best uniform approximation for $F\left(\tilde{\gamma}_{\text{o}}\right)$ always exists, and the error in this approximation gets smaller as we increase the number of the terms $l$. However, our main concern here is the relation defined by (\ref{eq3:completeMonotoneApprox}) more than finding the best number of terms or the unknown parameters $a_{i}$ and $k_{i}$. To conclude, both functions $S_{1}\left(\tilde{\gamma}_{\text{o}}\right)$ in (\ref{eq3:R-BPR1 sec}) and $S_{2}\left(\tilde{\gamma}_{\text{o}}\right)$ in (\ref{eq3:R-BPR2 sec}) can be approximated by a sum of exponential functions. We can assume that $l$ is large enough such that the approximation error in (\ref{eq3:completeMonotoneApprox}) is negligible.

Secondly, it is shown in \cite{shestopaloff2008sums} that the sum of exponential functions has at most two intersections with the abscissa. Consequently, and since the relation in (\ref{eq3:R-BPR sec}) can be expressed as a sum of exponential functions, the function $S\left(\tilde{\gamma}_{\text{o}}\right)$ has at most two roots in the interval $\left(0,+\infty \right)$. 
\end{proof}

\begin{theorem}
\label{theo3:th2}
There always exists a sufficiently small positive value $\epsilon$, such that $\epsilon \to 0^{+} $ and $\epsilon \ll \frac{\sigma_{i}^{2}}{N}$, $ \forall i \in [0,M]$, where the COPRA function $S\left(\tilde{\gamma_{o}}\right)$ in (\ref{eq3:R-BPR sec}) is zero (i.e., $\epsilon$ is a positive root for (\ref{eq3:R-BPR sec})). However, we are not interested in this root.
\end{theorem}

\begin{proof}

To start with, let $\epsilon = \tilde{\gamma_{o}}$, such that $\epsilon \to 0^{+} $ , $\epsilon \ll \frac{\sigma_{i}^{2}}{N}$, $ \forall i \in [0,M]$. As a result, equation~(\ref{eq3:p2_2}) can be written as  
 
\begin{align}
\label{ApB:eq1} 
S\left(\epsilon\right) &=
\sum_{i=1}^{M}\frac{\sigma_{i}^{2} b_{i}^{2}}{\left(\sigma_{i}^{2}+N\epsilon\right)^{2}}\Bigg[-4-\epsilon+  \sqrt{\epsilon} \sqrt{4+\epsilon}\Bigg]\nonumber\\
&+  
\sum_{i=1}^{M}\frac{ b_{i}^{2}}{\left(\sigma_{i}^{2}+N\epsilon\right)^{2}} \Bigg[-4N\epsilon+2 N\sqrt{\epsilon}\sqrt{\epsilon+4}-N\epsilon^{2} \nonumber\\
& +  N\epsilon \sqrt{\epsilon}\sqrt{\epsilon+4} \ \Bigg] .
\end{align}
Due to the properties of $\epsilon$, (\ref{ApB:eq1}) can be approximated as 
\begin{align}
\label{ApB:eq2} 
S\left(\epsilon\right) \approx -4\sum_{i=1}^{M}\sigma_{i}^{-2} b_{i}^{2} +
\sum_{i=1}^{M} \sigma_{i}^{-4}\ b_{i}^{2} \Big[&-4N\epsilon+4 N\sqrt{\epsilon} \nonumber\\
&+ 2 N\epsilon \sqrt{\epsilon} \Big].
\end{align}
Now, to simplify (\ref{ApB:eq2}), let us define $C_{1} = \sum_{i=1}^{M}\sigma_{i}^{-2} b_{i}^{2}$ and $C_{2} = \sum_{i=1}^{M} \sigma_{i}^{-4}\ b_{i}^{2} $. Substituting these two new variables in (\ref{ApB:eq2}) then manipulating result in
\begin{align}
\label{ApB:eq3} 
 &S\left(\epsilon\right) \approx 2 N C_{2} \epsilon \sqrt{\epsilon} +4 N C_{2} \sqrt{\epsilon}  -4N C_{2} \epsilon   -4C_{1}.
\end{align}
Solving $S\left(\epsilon\right)=0$ from (\ref{ApB:eq3}), we obtain one real root and two imaginary roots. The real root is given by  
\begin{equation}
\label{ApB:eq4}
\epsilon = \frac{\left(Q + \sqrt{Q^{2}+Z^{3}}\right)^{1/3}}{3\times 2^{1/3} N^{2} C_{2}^{2}}- \frac{4\times2^{1/3} N C_{2} \left(-2C_{1}+N C_{2}\right)}{\left(Q + \sqrt{Q^{2}+Z^{3}}\right)^{1/3}},
\end{equation}
where
\begin{equation}
\label{ApB:eq5}
Q = 108 N C_{1}C_{2},
\end{equation}
and
\begin{equation}
\label{ApB:eq6}
Z  = 19.05 \ N^{3} C_{2}^{3}\left(-2 C_{1} + N C_{2}\right).
\end{equation}
Now, we would like to know if this real root is positive or not. For (\ref{ApB:eq4}) to be positive, the following condition must hold 
\begin{align}
\label{ApB:eq7}
&\left(Q + \sqrt{Q^{2}+Z^{3}}\right)^{2/3} >  19.05 N^{3} C_{2}^{3} \left(-2 C_{1} +N C_{2}\right).
\end{align}
By using (\ref{ApB:eq6}), we can write (\ref{ApB:eq7}) as
\begin{align}
\label{ApB:eq71}
\left(Q + \sqrt{Q^{2}+Z^{3}}\right)^{2/3} > Z,
\end{align}
which can be easily proved. This concludes that $\epsilon$ is a positive real root for the COPRA function in (\ref{eq3:R-BPR sec}).

Secondly, we would like to know if $\epsilon$ can be considered as a value for our regularization parameter $\tilde{\gamma_{\text{o}}}$. A direct way to prove that can be noted from the fact that having $\epsilon \ll \frac{\sigma_{i}^{2}}{N}$, $ \forall i \in [0,M]$ will not provide any source of regularization to the problem. Hence, the RLS solution converges to the LS.

As a remark, we can assume that the approximation in (\ref{ApB:eq2}) is uniform such that it does not affect the position of the roots. Thus, we can claim that this root is not coming from the negative region of the $x$ axis. However, we can easily prove that the function in (\ref{ApB:eq1}) does not have a negative real root that is close to zero (at least in the region from 0 to -2). Thus, this root is not coming from the negative region as a result of approximating the function (i.e., perturbed root).

\end{proof}

\begin{theorem}
\label{theo3:th4}
The COPRA function $S\left(\tilde{\gamma}_{\text{o}}\right)$ has a unique positive root in the interval $\left(\epsilon, +\infty \right)$.
\end{theorem}
\begin{proof}

According to Theorem~\ref{theo3:th1}, the function $S\left( \tilde{\gamma}_{\text{o}} \right)$ can have no root, one, or two roots. However, we have already proved in Theorem~\ref{theo3:th2} that there exists a significantly small positive root for COPRA function at $\tilde{\gamma}_{\text{o,1}} =\epsilon$ but we are not interested in this root.  In other words, we would like to see if there exists a second root for $S\left(\tilde{\gamma}_{\text{o}}\right)$ in the interval $\left(\epsilon, +\infty\right)$. 

Property~\ref{pro3:p2} shows that the COPRA function starts from a negative value, whereas Property~\ref{pro3:p4} states that the COPRA function approaches zero at $+\infty$ from a negative value. This means that the COPRA function has a negative value before $\epsilon$, then it switches to the positive region after that. Since Property~\ref{pro3:p4} guarantees that the COPRA function approaches zero from a negative direction as $\tilde{\gamma}_{\text{o}}$ approaches $+\infty$, then $S\left( \tilde{\gamma}_{\text{o}} \right)$ has an extremum in the interval $\left(\epsilon, +\infty\right)$ and this extremum is actually a maximum point. If the point of the extremum is considered to be $\tilde{\gamma}_{\text{o,m}}$, then the function starts decreasing for $\tilde{\gamma}_{\text{o}} > \tilde{\gamma}_{\text{o,m}}$ until it approaches the second zero crossing at $\tilde{\gamma}_{o,2}$. As Theorem~\ref{theo3:th1} states clearly that we cannot have more than two roots, we conclude that the COPRA function in (\ref{eq3:R-BPR sec}) has only one unique positive root over the interval $\left(\epsilon, +\infty\right)$.

\end{proof}
\subsection{Finding the root of $S\left(\tilde{\gamma}_{\text{o}}\right)$}
\label{subsec2:find root}
To find the positive root of the COPRA function $S\left(\tilde{\gamma}_{\text{o}}\right)$ in (\ref{eq3:R-BPR sec}), Newton's method \cite{zarowski2004introduction} can be used. The function $S\left(\tilde{\gamma}_{\text{o}}\right)$ is differentiable in the interval $\left(\epsilon, +\infty\right)$ and the expression of the first derivative $S^{'}\left(\tilde{\gamma}_{\text{o}}\right)$ can be easily obtained. Newton's method can then be applied in a straightforward manner to find this root. Starting from an initial value $\gamma_{\text{o}}^{n=0} > \epsilon $ that is sufficiency small, the following iterations are performed:
 \begin{equation}
\label{eq:Newton}
\tilde{\gamma}_{\text{o}}^{n+1} =\tilde{\gamma}_{\text{o}}^{n} - \frac{S\left(\tilde{\gamma}_{\text{o}}\right)}{S^{'}\left(\tilde{\gamma}_{\text{o}}\right) }.
\end{equation}
The iterations stop when $|S(\tilde{\gamma}_{\text{o}}^{n+1})|< \rho$, where $\rho$ is a sufficiently small positive quantity.

\subsection{Convergence}
\label{subsec3:converge}
The convergence of Newton's method can be easily proved. As a result from Theorem~\ref{theo3:th4}, the function $S\left(\tilde{\gamma}_{\text{o}}\right)$ has always a positive value in the interval $\left(\epsilon, \tilde{\gamma}_{o,2}\right)$. It is also clear that $S\left(\tilde{\gamma}_{\text{o}}\right)$ is a decreasing function in the interval $[\tilde{\gamma}_{\text{o}}^{n= 0}, \tilde{\gamma}_{\text{o,2}}]$. Thus, starting from $\tilde{\gamma}_{\text{o}}^{n= 0} \gg \epsilon$, (\ref{eq:Newton}) will produce a consecutive increase estimation for $\tilde{\gamma}_{\text{o}}$. Convergence occurs when $S\left(\tilde{\gamma}_{\text{o}}^{n}\right) \rightarrow 0$ and $\tilde{\gamma}_{\text{o}}^{n+1}\rightarrow \tilde{\gamma}_{\text{o}}^{n}$. 
\subsection{COPRA summery}
\label{subsec3:R-BPR summery}
The proposed COPRA can be summarized as in Algorithm~\ref{RCOPRA ALGORITHM}.

\begin{algorithm}[h!]
\caption{COPRA Summery}
\begin{algorithmic} [1]
\State Define $\rho$ as the iterations stopping criterion.
\State Set $\tilde{\gamma}_{\text{o}}^{n =0}$ to be a sufficiently small positive quantity.
\State Find {$S\left(\tilde{\gamma}_{\text{o}}^{n=0}\right)$} using (\ref{eq3:R-BPR sec}), and compute its derivative {$S^{'}\left(\tilde{\gamma}_{\text{o}}^{n=0}\right)$.}
\While{$|S\left(\tilde{\gamma}_{\text{o}}^{n}\right)| > \rho$}
\State Solve (\ref{eq:Newton}) to get $\tilde{\gamma}_{\text{o}}^{n+1}$.
\State $\tilde{\gamma}_{\text{o}}^{n} =\tilde{\gamma}_{\text{o}}^{n+1}$.
\EndWhile
\State $\gamma_{\text{o}} = N \tilde{\gamma}_{\text{o}}$.
\State Find $\hat{\xv}$ using $\hat{\xv}= \left(\Hm^{H}\Hm+\gamma_{\text{o}}\Id\right)^{-1}\Hm^{H}\yv$.
\end{algorithmic}
\label{RCOPRA ALGORITHM}
\end{algorithm}

\section{Extension: Numerical results}

In this section we present an additional simulation results for the proposed COPRA. 

Firstly, the model matrix is generated as $\Hm \in \mathbb{C}^{100\times90}, \Hm \sim \mathcal{CN}\left(\bm{0}, \Id\right)$ with i.i.d entries. This $\Hm$ is combined with $\xv \sim \mathcal{N}(\bm{0}, \Id)$ that has an i.i.d. elements. The performance is evaluated in terms of \emph{normalized} MSE (NMSE) versus signal-to-noise-ratio (SNR) defined as SNR $ \triangleq  \lVert \Hm\xv\rVert_2^{2}/N \sigma_{\zv}^{2}$. The performance of all the methods is evaluated over $10^{5}$ different noise and matrix realizations at each SNR. From Fig~\ref{fig:under}, it can be seen that the proposed COPRA is providing a NMSE that is very close to the LMMSE estimator. GCV  algorithm is also providing a close performance to the proposed COPRA. For a very high SNR, all methods are providing approximately the same NMSE. 

\begin{figure}[h!]
  \centerline{\includegraphics[width=  3.8in, height = 2.9in]{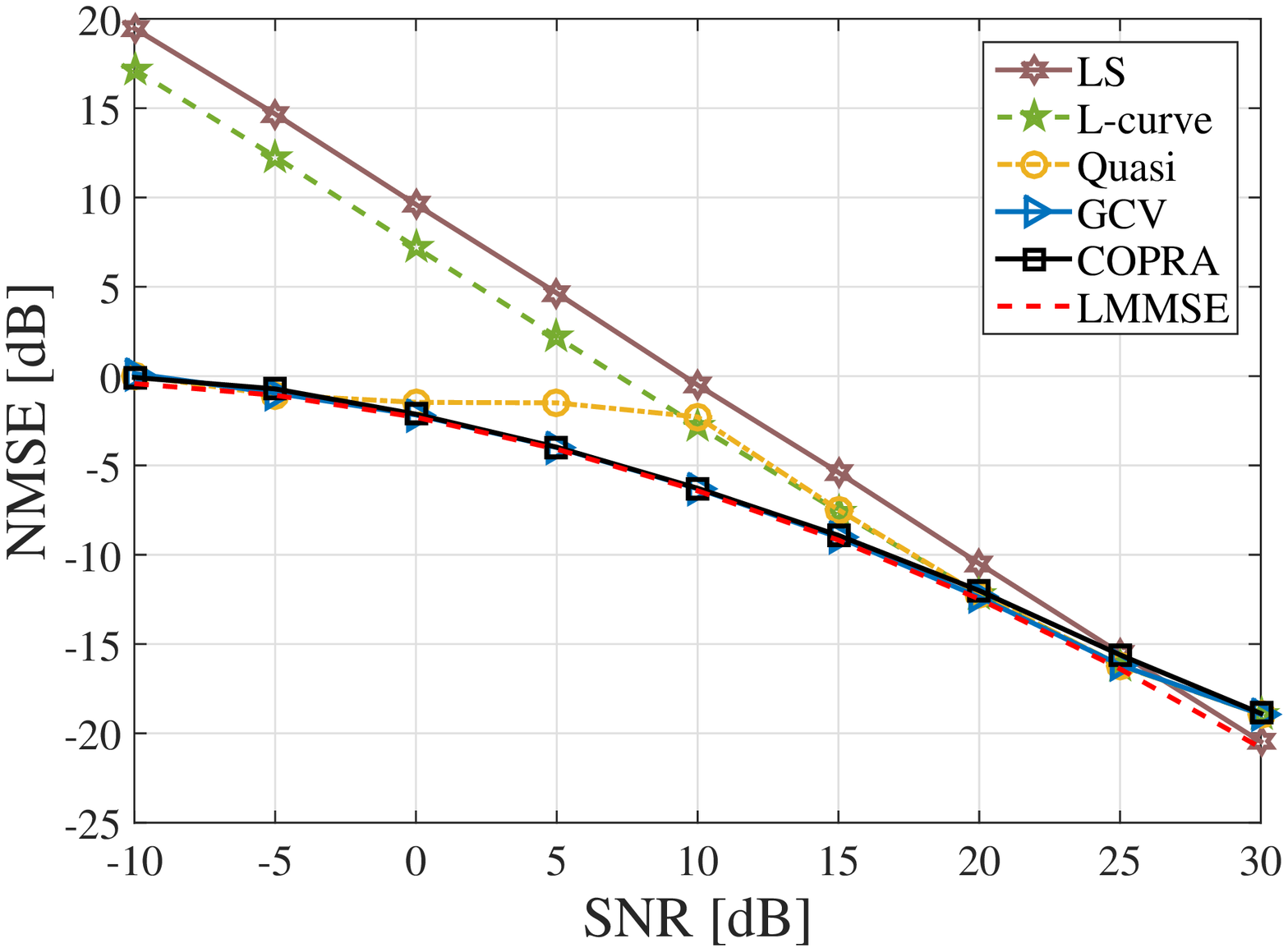}}  
\caption{Performance comparison when $\Hm \in \mathbb{C}^{100\times90}, \Hm \sim \mathcal{CN}\left(\bm{0}, \Id\right)$ with i.i.d entries and $\xv \sim \mathcal{N}(\bm{0}, \Id)$ with i.i.d. elements.}
\label{fig:under}
\end{figure}
Secondly, the matrix $\Hm$ is generated as $\Hm \in \mathbb{R}^{100\times100}, \Hm \sim \mathcal{N}\left(\bm{0}, \Id\right)$ that has an i.i.d entries. This $\Hm$ is combined with a stochastic Gaussian signal $\xv$ that has an independent but not identically distributed (i.n.d.) entries of zero mean and unit variance.  The performance is presented as the NMSE versus SNR (in dB) and is evaluated over $10^5$ different noise and $\Hm$ realizations at each SNR value. 

From Fig~\ref{fig:ind}, it can be observed that the proposed COPRA outperforms all the benchmark methods over all the SNR values. It is also clear that the LS performance it above 10 dB overall the SNR range, which makes it completely unreliable.

\begin{figure}[h!]
  \centerline{\includegraphics[width=  3.8in, height = 2.9in]{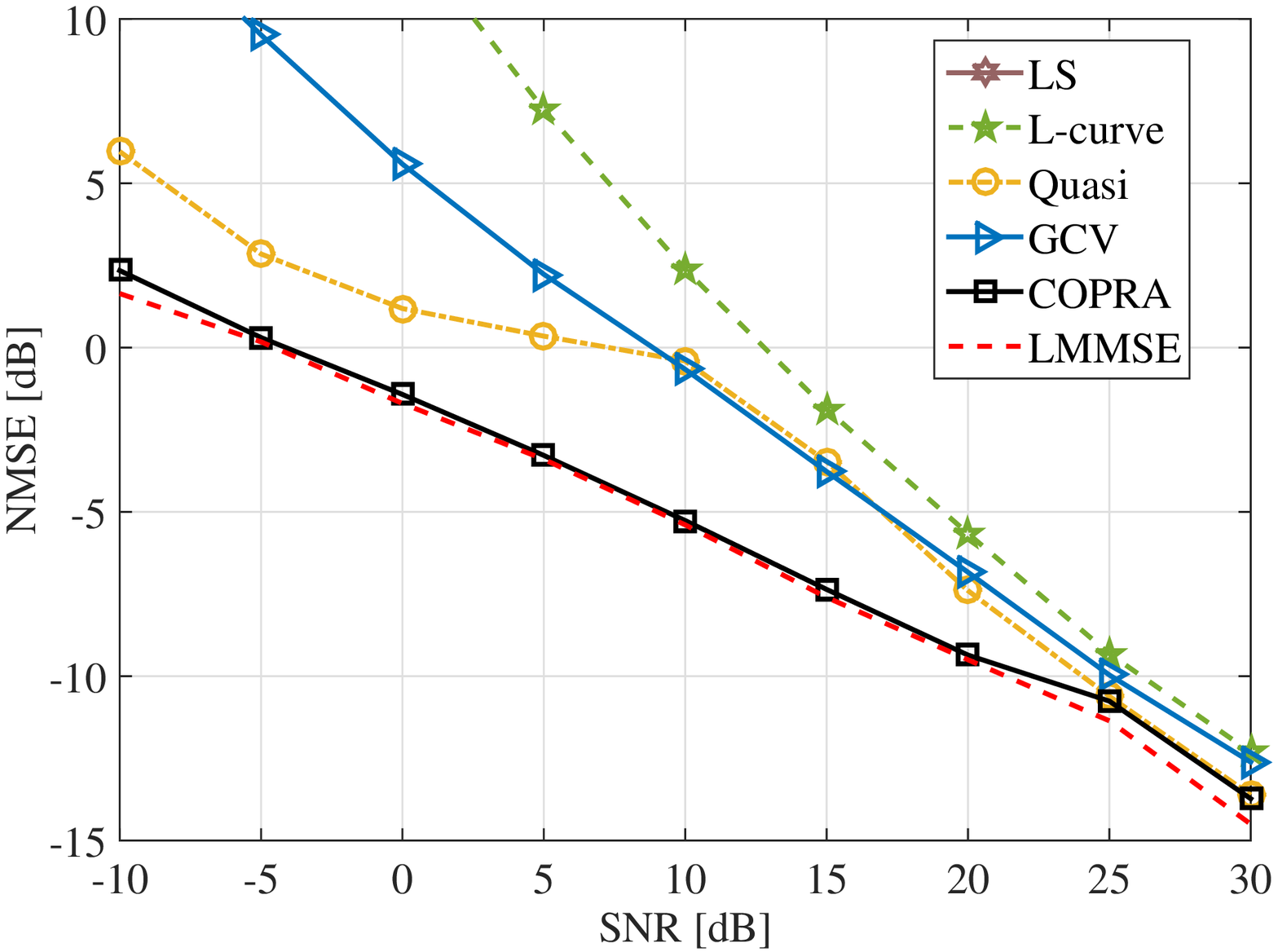}}  
\caption{Performance comparison when $\Hm \in \mathbb{R}^{100\times100}, \Hm \sim \mathcal{N}\left(\bm{0}, \Id\right)$ with i.i.d entries and $\xv$ is stochastic Gaussian with i.n.d. elements.}
\label{fig:ind}
\end{figure}

Finally, a different scenario of the transmitted signal is considered where $\xv$ is taken to be a gray encoded 8-ary QAM signal with unit power and the model matrix is generated as $\Hm  \in \mathbb{C}^{100\times100}, \Hm \sim \mathcal{CN}(\bm{0}, \Id)$. In this example, the performance is evaluated in terms of the bit error rate (BER), and also the NMSE. Noise is added to $\Hm \xv$ according to a certain $\text{E}_{\text{b}}/\text{N}_{\text{o}}$ (energy per bit to noise power spectral density ratio) (in dB) to generate $\yv$. Performance is presented as the BER versus  $\text{E}_{\text{b}}/\text{N}_{\text{o}}$ (in dB), and also the NMSE (in dB) versus $\text{E}_{\text{b}}/\text{N}_{\text{o}}$ (in dB). The performance of all the methods is evaluated over $10^5$ different realizations of the noise $\zv$ and the matrix $\Hm$ at each $\text{E}_{\text{b}}/\text{N}_{\text{o}}$ value.

Fig.~\ref{fig:qam1} plots the BER versus $\text{E}_{\text{b}}/\text{N}_{\text{o}}$ (in dB). It is evident that the proposed COPRA outperforms all the benchmark methods and stay very close to the LMMSE especially in the range from 0 to 15~dB $\text{E}_{\text{b}}/\text{N}_{\text{o}}$ where it offers exactly the BER provided by the LMMSE estimator. 

Fig.~\ref{fig:qam2} depicts the comparison between the methods in terms of the NMSE. Again, we can observe that the proposed COPRA outperforms all the methods and performs almost identically to the LMMSE estimator in the low $\text{E}_{\text{b}}/\text{N}_{\text{o}}$ range.
\begin{figure}[h!]
\centering
 \begin{subfigure}[h]{0.4\textwidth}   
\centerline{\includegraphics[width=  3.8in, height = 2.9in]{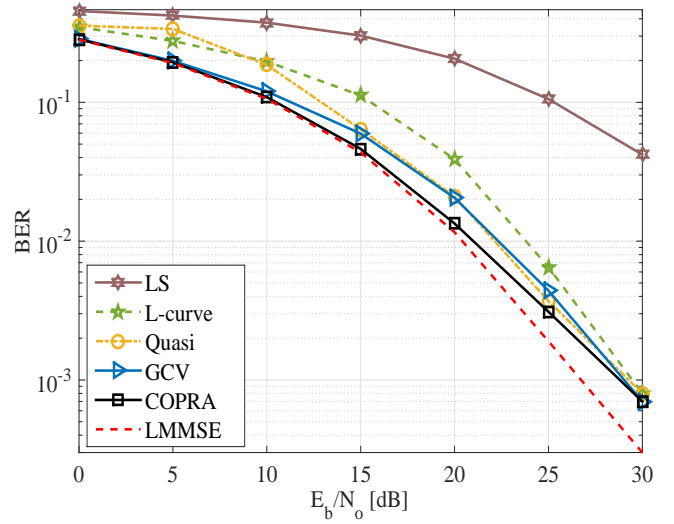}} 
\caption{BER versus $E_{\text{b}}/N_{\text{o}}$ [dB].}
\label{fig:qam1}
    \end{subfigure}%
  
    \begin{subfigure}[h]{0.4\textwidth}
\centerline{\includegraphics[width=  3.8in, height = 2.9in]{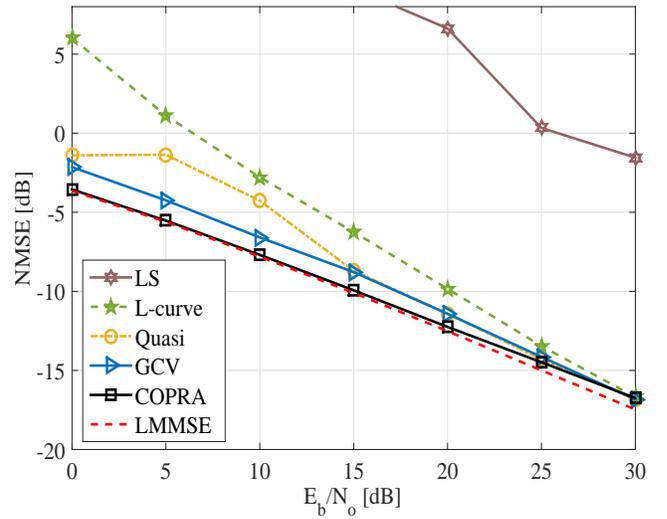}}  
\caption{NMSE [dB] versus $E_{\text{b}}/N_{\text{o}}$ [dB].}
\label{fig:qam2}
    \end{subfigure}
\caption{Performance comparison when $\Hm  \in \mathbb{C}^{100\times100}, \Hm \sim \mathcal{CN}(\bm{0}, \Id)$ is combined with gray encoded 8-ary QAM signal with unit power.}
\label{fig:qam}
\end{figure}

\newpage 
\bibliographystyle{IEEEbib}
\bibliography{Appendex_double}

\end{document}